\documentclass[%
 reprint,
 amsmath,amssymb,
 aps,
]{revtex4-2}

\usepackage{graphicx}
\usepackage{dcolumn}
\usepackage{bm}
\usepackage[breaklinks]{hyperref}
\usepackage{braket}
\usepackage{color}
\def\tr{\operatorname{tr}}
\renewcommand{\(}{\left(}
\renewcommand{\)}{\right)}

\newcommand{\Eqref}[1]{Eq.~\eqref{#1}}
\newcommand{\secref}[1]{Sec.~(\ref{#1})}
\newcommand{\figref}[1]{Fig.~(\ref{#1})}

\usepackage{amsthm}
\theoremstyle{plain}
\newtheorem{thm}{Theorem}
\newtheorem{lem}{Lemma}

\theoremstyle{definition}
\newtheorem{defn}{Definition}
\newtheorem{rem}{Remark}

\begin{document}

\preprint{APS/123-QED}

\title{Entanglement of Purification in Random Tensor Networks}

\author{Chris Akers}
\email{cakers@mit.edu}
\affiliation{Center for Theoretical Physics,\\
Massachusetts Institute of Technology, Cambridge, MA 02139, USA
}%
\author{Thomas Faulkner}%
 \email{tomf@illinois.edu}
\author{Simon Lin}
\email{shanlin3@illinois.edu}
\affiliation{Department of Physics, University of Illinois,\\ 1110 W. Green St., Urbana, IL 61801-3080, USA
}%
\author{Pratik Rath}
\email{rath@ucsb.edu}
\affiliation{Department of Physics, University of California,\\
Santa Barbara, CA 93106, USA
}%

\begin{abstract}
The entanglement of purification $E_P(A\colon  B)$ is a powerful correlation measure, but it is notoriously difficult to compute because it involves an optimization over all possible purifications. 
In this paper, we prove a new inequality: $E_P(A\colon  B)\geq \frac{1}{2}S_R^{(2)}(A\colon  B)$, where $S_R^{(n)}(A\colon  B)$ is the Renyi reflected entropy. 
Using this, we compute $E_P(A\colon  B)$ for a large class of random tensor networks at large bond dimension and show that it is equal to the entanglement wedge cross section $EW(A\colon  B)$, proving a previous conjecture motivated from AdS/CFT. 
\end{abstract}

\maketitle

\section{Introduction}
Given a bipartite density matrix $\rho_{AB}$, the entanglement of purification $E_P(A\colon  B)$ is defined as \cite{terhal2002entanglement}
\begin{equation}
E_P(A\colon  B) = \min_{\ket{\psi}_{ABA'B'}} S(AA'),
\end{equation}
where $S(R)=-\tr \left(\rho_R \log \rho_R \right)$ is the von Neumann entropy. 
The minimization runs over all possible purifications of $\rho_{AB}$, i.e., $\ket{\psi}_{ABA'B'}$ such that $\tr_{A'B'} \left(\ket{\psi} \bra{\psi} \right)=\rho_{AB}$, and the $\ket{\psi}$ that achieves the minimum is called the optimal purification. 
$E_P(A\colon  B)$ is a useful measure of correlations in a bipartite mixed state and is proven to be monotonic under local operations \cite{terhal2002entanglement}. 
However, it is generally intractable to compute because of the optimization over all possible purifications \footnote{Exceptions to this include pure states like Bell pairs and classically correlated states like GHZ states, see Ref.~\cite{Nguyen:2017yqw} for details.}.

In the context of AdS/CFT \footnote{See Ref.~\cite{Harlow:2018fse} for a review of the quantum information perspective on AdS/CFT.}, it has been conjectured that for $A, B$ subregions of the CFT, there is a simple geometric, AdS dual to $E_P(A:B)$.
The entanglement wedge of subregion $AB$ of the CFT is the bulk region between $AB$ and the minimal surface $\gamma_{AB}$ (also called the Ryu-Takayanagi (RT) surface \cite{Ryu:2006bv}).
This is, in appropriate settings, the bulk region reconstructable from the corresponding boundary subregion \cite{Dong:2016eik}.
Based on this, Refs.~\cite{Takayanagi:2017knl,Nguyen:2017yqw} conjectured that $E_P(A\colon  B)$ is given by
\begin{equation}\label{eq:conj}
    E_P(A\colon B)=EW(A\colon B)=\frac{\text{Area}(\Gamma_{A:B})}{4G_N},
\end{equation}
where $\Gamma_{A:B}$ is the entanglement wedge cross section, the minimal surface dividing the entanglement wedge into portions containing $A$ and $B$ respectively, as depicted in \figref{fig:EW}.
$G_N$ is Newton's constant and in this paper, we will set $\hbar=c=1$ by choosing natural units. 

\begin{figure}[t]
\includegraphics[scale=0.5]{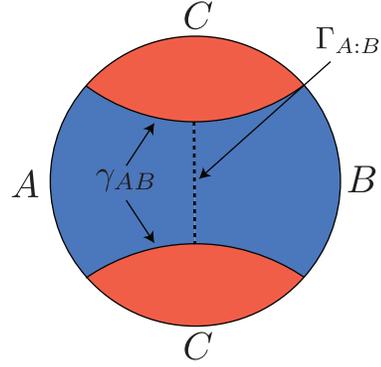}
\caption{The entanglement wedge cross section $\Gamma_{A:B}$ divides the entanglement wedge (blue) into two regions which contain $A$ and $B$ respectively. The RT surface is denoted $\gamma_{AB}$.}\label{fig:EW}
\end{figure}


Proving this AdS/CFT conjecture appears quite challenging.
However, there exists a toy model of AdS, called random tensor networks (RTNs), which have proven useful in discovering new insights into AdS/CFT entanglement properties \cite{Dong:2021clv,Akers:2020pmf,Akers:2021pvd,Akers:2022zxr,akers2023reflected}, especially because of their connection to fixed-area states \cite{Akers:2018fow,Dong:2018seb,Dong:2019piw,Dong:2022ilf}.
The goal of this note is to present progress on proving the conjecture \eqref{eq:conj} in RTNs.

We compute $E_P$ by using a known upper bound and deriving a new lower bound (Theorem~\ref{thm:lower_bound}), which we are able to argue matches the upper bound in certain RTNs.
This argument relies on results obtained previously for the reflected entropy, $S_R(A\colon B)$, in RTNs \cite{Akers:2021pvd,Akers:2022zxr,akers2023reflected}.
The reflected entropy is defined as \cite{Dutta:2019gen}
\begin{equation}
S_R(A\colon B)=S(AA^*)_{\ket{\sqrt{\rho_{AB}}}},
\end{equation}
where the state $\ket{\sqrt{\rho_{AB}}}$ is the canonical purification, which lives in the Hilbert space $\mathrm{End}\(\mathcal{H}_{AB}\)$ of operators acting on $\mathcal{H}_{AB}$.
$\mathrm{End}\(\mathcal{H}_{AB}\)$ is isomorphic to the doubled Hilbert space $\mathcal{H}_{AB}\otimes\mathcal{H}_{A^*B^*}$.

The bounds are as follows.
It is conjectured that the reflected entropy in AdS/CFT satisfies
\begin{equation}
	S_R(A:B) = 2 EW(A:B),
\end{equation} 
and this has been proven rigorously for a large class of RTNs \cite{Akers:2021pvd,Akers:2022zxr,akers2023reflected}, as we will discuss.
Moreover, as argued in \cite{Takayanagi:2017knl}, RTNs in general satisfy
\begin{equation}
	E_P(A:B) \le EW(A:B).
\end{equation}
This places the upper bound $E_P \le S_R / 2$.
The rest of this paper proves the lower bound and discusses when it matches this upper bound.

\section{Reflected Entropy from Modular Operator}
\label{sec:mod}

\begin{defn}
The Renyi reflected entropy is
\begin{align}
    S_R^{(n)}(A\colon B)=S_n(AA^*)_{\ket{\sqrt{\rho_{AB}}}},\label{eq:SRn}
\end{align}
where $S_n(R)=\frac{1}{1-n}\log \tr \left(\rho_R^n\right)$ is the $n$th Renyi entropy. 
\end{defn}

The lower bound in Theorem~\ref{thm:lower_bound} will require the following lemma
that 
rewrites the Renyi reflected entropy using the formalism of modular operators appearing in Tomita-Takesaki theory \footnote{See Ref.~\cite{Witten:2018zxz} for a review.}.
Consider a finite dimensional system with Hilbert space $\mathcal{H}_{AB} \otimes \mathcal{H}_{C}$, where subsystem $C$ is completely general.  Given a state $\ket{\psi}$ \footnote{$\ket{\psi}$ does not need to be cyclic and separating.}  and subsystem $AB$, the modular operator is defined as
\begin{equation}
    \Delta_{AB,\psi} = \rho_{AB}\otimes \rho_{C}^{-1},
\end{equation}
where the inverse is defined to act only on the non-zero subspace of $\rho_{C}$ and $\Delta_{AB,\psi}$ is defined to annihilate the orthogonal subspace. 

\begin{lem}
For integer $n \geq 2$, 
\begin{equation}
	S_R^{(n)}(A\colon B) = \frac{1}{1-n}\log \bra{\psi^{\otimes n}}\Sigma_{A}\Delta_{AB^{\otimes n},\psi^{\otimes n}}^{1/2}\Sigma_{A}^{\dagger}\ket{\psi^{\otimes n}}~,
\end{equation}
where $\Sigma_{A(A^*)}$ are twist operators that cyclically permute the $n$ copies of $\ket{\sqrt{\rho_{AB}}}$ on subregion $A(A^*)$, $\ket{\psi}$ is an arbitrary purification of $\rho_{AB}$, and $\Delta_{AB^{\otimes n},\psi^{\otimes n}}=\Delta_{AB,\psi}^{\otimes n}$.
\end{lem}

\begin{proof}

Start with \Eqref{eq:SRn} and rewrite it as \cite{Dutta:2019gen}
\begin{align}\label{eq:SRntwist1}
    S_R^{(n)}(A\colon B) &=\frac{1}{1-n}\log \tr\(\rho_{AA^*}^n\)
    \\
    \tr\(\rho_{AA^*}^n\) &= \bra{\sqrt{\rho_{AB}}^{\otimes n}}\Sigma_{A}\Sigma_{A^*}\ket{\sqrt{\rho_{AB}}^{\otimes n}}.\label{eq:SRntwist}
\end{align}
As described in Ref.~\cite{Dutta:2019gen}, operators act on $\mathrm{End}\(\mathcal{H}_{AB}\)$ by left and right actions, i.e.,
\begin{align}
    O_{AB}\ket{M_{AB}} &= \ket{O_{AB}\,M_{AB}} \\
    O_{A^*B^*}\ket{M_{AB}} &= \ket{M_{AB}\,O_{AB}^{\dagger}},
\end{align}
and the inner product is defined by
\begin{equation}
    \braket{M|N}=\tr\(M^{\dagger}N\).
\end{equation}
Using this, one finds that \Eqref{eq:SRntwist} is given by
\begin{equation}\label{eq:SRtrace}
     \tr\(\rho_{AA^*}^n\) = \tr_{(AB)^{\otimes n}}\(\sqrt{\rho_{AB}}^{\otimes n}\Sigma_{A}\sqrt{\rho_{AB}}^{\otimes n}\Sigma_{A}^{\dagger}\).
\end{equation}

To express \Eqref{eq:SRtrace} in terms of modular operators, we consider an arbitrary purification of $\rho_{AB}$ denoted $\ket{\psi}$, giving
\begin{equation}
\label{eq:deltaSR}
\begin{split}
    \tr\(\rho_{AA^*}^n\)&= \tr_{(AB)^{\otimes n}}\(\sqrt{\rho_{AB}}^{\otimes n}\Sigma_{A}\sqrt{\rho_{AB}}^{\otimes n}\Sigma_{A}^{\dagger}\)\nonumber \\
    &=\bra{\psi^{\otimes n}}\Sigma_{A}\Delta_{AB^{\otimes n},\psi^{\otimes n}}^{1/2}\Sigma_{A}^{\dagger}\Delta_{AB^{\otimes n},\psi^{\otimes n}}^{-1/2}\ket{\psi^{\otimes n}}\nonumber\\
    &=\bra{\psi^{\otimes n}}\Sigma_{A}\Delta_{AB^{\otimes n},\psi^{\otimes n}}^{1/2}\Sigma_{A}^{\dagger}\ket{\psi^{\otimes n}},
\end{split}
\end{equation}
where we have used the fact that the $\rho_C$ dependence cancels out in the second line. 
For the last line, we have used $\Delta_{AB,\psi}^{-1/2}\ket{\psi}=\ket{\psi}$ which is easy to see by working in the Schmidt basis. 
\end{proof}

\section{Lower bound}
\label{sec:ineq}

\begin{thm}\label{thm:lower_bound}
For integer $n \ge 2$,
\begin{equation}\label{eq:ineq_main}
	E_P(A:B) \ge S^{(n)}_R(A:B) / 2.
\end{equation}
\end{thm}

\begin{rem}
In Ref.~\cite{Dutta:2019gen}, it was proven that for integer $n \ge 2$, the Renyi reflected entropy is monotonic under partial trace, i.e., $S_R^{(n)}(A\colon  BC)\geq S_R^{(n)},(A\colon  B)$. 
This immediately implies Theorem \ref{thm:lower_bound} by the following argument.
Let $\ket{\psi}_{ABA'B'}$ be the optimal purification.
Then
\begin{equation}
    2S(AA')\geq 2S_n(AA')=S_R^{(n)}(AA'\colon  BB')\geq  S_R^{(n)}(A\colon  B),
\end{equation}
where we have used the fact that $S^{(n)}_R(C:D)=2S_n(C)$ for a pure state on $CD$.
That said, we choose to present the proof below because it is self-contained and far simpler than the proof of monotonicity in Ref.~\cite{Dutta:2019gen}.
\end{rem}

\begin{proof}[Proof of Theorem \ref{thm:lower_bound}]

We first define the Renyi generalization of $E_P(A\colon B)$ as
\begin{align}\label{eq:EPn}
    E_P^{(n)}(A\colon B) = \min_{\ket{\psi}_{ABA'B'}} S_n(AA').
\end{align}
Applying the monotonicity of Renyi entropy, i.e., $\partial_n S_n \leq 0$, for $n>1$ we have 
\begin{equation}\label{eq:renyiEP}
E_P(A\colon B)\geq E_P^{(n)}(A\colon B).
\end{equation}

Now consider an arbitrary purification $\ket{\psi}_{ABA'B'}$. For integer $n\geq 2$, the Renyi entropy for subregion $AA'$ can be computed using twist operators in a fashion similar to Eqs.~(\ref{eq:SRntwist1},\ref{eq:SRntwist}), i.e., 
\begin{align}
    S_n(AA')&=\frac{1}{1-n}\log \tr\(\rho_{AA'}^n\)\\
    \tr\(\rho_{AA'}^n\) &=\bra{\psi^{\otimes n}}\Sigma_{A}\Sigma_{A'}\ket{\psi^{\otimes n}}.\label{eq:twist}
\end{align}

Define the operators $\Pi_{AB,\psi}$ ($\Pi_{A'B',\psi}$) to be projectors onto the non-zero subspaces of the reduced density matrices on $AB$ ($A'B'$). Then, using $\Pi_{AB,\psi}\ket{\psi}=\Pi_{A'B',\psi}\ket{\psi}=\ket{\psi}$, we can insert $\Pi_{AB,\psi}$ ($\Pi_{A'B',\psi}$) from the right (left) in \Eqref{eq:twist} for each of the $n$ copies of $\ket{\psi}$. Note that $\Pi_{AB}\Pi_{A'B'}=\Delta_{AB,\psi}^{1/4}\Delta_{AB,\psi}^{-1/4}$ as the inverse density matrices in the modular operators annihilate the orthogonal subspaces. We can use this fact to insert a pair of modular operators into \Eqref{eq:twist} to get
\begin{equation}
\begin{split}
\tr\left(\rho_{AA'}^n\right)&=\bra{\psi^{\otimes n}}\Sigma_{A}\(\Delta_{AB,\psi}^{1/4}\Delta_{AB,\psi}^{-1/4}\)^{\otimes n}\Sigma_{A'}\ket{\psi^{\otimes n}}
    \\&\leq \left(\bra{\psi^{\otimes n}}\Sigma_{A}\Delta_{AB^{\otimes n},\psi^{\otimes n}}^{1/2}\Sigma_{A}^\dagger \ket{\psi^{\otimes n}}\right.\\
    &~~~~\left.\bra{\psi^{\otimes n}}\Sigma_{A'}\Delta_{AB^{\otimes n},\psi^{\otimes n}}^{-1/2}\Sigma_{A'}^\dagger\ket{\psi^{\otimes n}}\right)^{\frac{1}{2}},
\end{split}\label{eq:ineq}
\end{equation}
where we have applied the Cauchy-Schwarz inequality between the modular operators. 

Using $\Delta_{AB,\psi}^{-1}=\Delta_{A'B',\psi}$ and \Eqref{eq:deltaSR}, the two terms in the last line of \Eqref{eq:ineq} can be related to Renyi reflected entropies on $A:B$ and $A':B'$ respectively. Thus, we have
\begin{equation}\label{eq:ineq1}
    2\frac{1}{1-n} \log \tr\left(\rho_{AA'}^n\right) \geq S_R^{(n)}(A\colon  B) +S_R^{(n)}(A'\colon  B').
\end{equation}
Finally using the fact that $S_R^{(n)}(A'\colon  B')\geq 0$, applying \Eqref{eq:ineq1} to the optimal purification arising in the calculation of $E_P^{(n)}(A\colon  B)$ and using \Eqref{eq:renyiEP}, we have our desired inequality.
\end{proof}

\begin{rem}
    We will use the inequality at $n=2$ since it is the strongest.
\end{rem}

\begin{rem}
It is important to note that this inequality was derived using twist operators which only exist at integer $n$. 
In the context of computing entanglement entropy, one usually analytically continues the answer obtained at integer $n$ to non-integer values using Carlson's theorem. 
However, it is not necessarily possible to analytically continue an inequality. 
For example, the monotonicity of Renyi reflected entropy under partial trace, i.e., $S_R^{(n)}(A\colon  BC)\geq S_R^{(n)},(A\colon  B)$ was proved to be true at integer $n$ \cite{Dutta:2019gen}, whereas counterexamples were found for non-integer $n$ in Ref.~\cite{Hayden:2023yij}.
\end{rem}

\section{Random Tensor Networks}
\label{sec:SR}

We can now use these bounds to compute $E_P$ in many random tensor network states.
These states are defined as (up to normalization) \cite{Hayden:2016cfa}
\begin{equation}
    \ket{\psi}=\(\prod_{<xy>\in E}\bra{xy}\)\(\prod_{x\in V} \ket{V_x}\),
\end{equation}
where we are considering an arbitrary graph defined by vertices $V$ and edges $E$. 
The states $\ket{V_x}$ are Haar random and the states $\ket{xy}$ are maximally entangled. 
This defines a state on the vertices living at the boundary of the graph. 
We will consider RTNs in the simplifying limit where all bond dimensions $\chi_{xy}$ are large such that $\log \chi_{xy} \propto \log D$ and $D\to \infty$ \footnote{$\log D\sim \frac{1}{4G_N}$ in AdS/CFT in units where $l_{AdS}=1$.}.

\begin{figure}[t]
\includegraphics[scale=0.5]{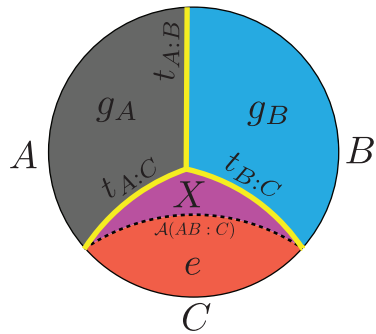}
\caption{The triway cut (yellow) minimizes the energy cost of the domain walls with tensions $t_{A:B}=1$ and $t_{A:C}=t_{B:C}=\frac{n}{2(n-1)}$. For $n>1$, it lies within the entanglement wedge of $AB$ defined by the RT surface denoted $\mathcal{A}(AB:C)$. The optimal configuration corresponds to domains of permutation elements $\{e,g_A,g_B,X\}$ as shown.}\label{fig:triway}
\end{figure}

For RTN states, the Renyi reflected entropy is computed by finding the optimal configuration of permutations that minimizes a certain free energy (see Ref.~\cite{akers2023reflected} for details). 
It was proved in Ref.~\cite{akers2023reflected} that the optimal configuration involves four permutation elements $\{e,g_A,g_B,X\}$ and takes the general form shown in \figref{fig:triway}.
In detail, we have
\begin{equation}\label{eq:triway}
    \lim_{D\to \infty}\frac{S_R^{(n)}(A\colon B)}{\log D} = 2 \mathcal{A}_n(A\colon  B\colon  C) - \frac{n}{n-1} \mathcal{A}(AB\colon  C),
\end{equation}
where $\mathcal{A}_n(A\colon  B\colon  C)$ is the triway cut with tensions $t_{A:B}=1$ and $t_{A:C}=t_{B:C}=\frac{n}{2(n-1)}$ (see \figref{fig:triway}). $\mathcal{A}(AB\colon  C)$ is the minimal cut separating $AB$ from $C$.

While the triway cut problem provides a natural analytic continuation in $n$ and Refs.~\cite{Akers:2021pvd,Akers:2022zxr} have provided evidence that this in fact is the correct prescription, it is not necessary to assume this for the purpose of this paper. 
For now we note that at $n=2$, all the tensions are equal and normalized to 1.
On the other hand, in the limit $n\to 1$, the RHS of \Eqref{eq:triway} approaches $2EW(A:B)$.  

Now, the key point is that there exist networks where the triway cut configuration is identical for $n\to1$ and $n=2$. 
This corresponds to networks where the $X$ region in \figref{fig:triway} vanishes at $n=2$. 
We will demonstrate such examples in \secref{sec:eg}. 
For now, assuming such a network and using \Eqref{eq:ineq_main}, we have 
\begin{equation}
    E_P(A\colon  B)\geq \frac{1}{2}S_R^{(2)}(A\colon  B) = EW(A\colon  B).
\end{equation}

To prove the opposite inequality, we repeat the arguments made in Refs.~\cite{Takayanagi:2017knl,Nguyen:2017yqw}. 
There is an approximate isometry relating the RTN state $\ket{\psi}_{ABC}$ to the state $\ket{\psi}_{ABC'}$ defined on the same graph truncated to the entanglement wedge of $AB$, with $C'=\gamma_{AB}$. 
The RT formula can still be applied and optimizing over the choice of decomposition $C=A'\cup B'$, we have $S(AA')=EW(A\colon  B)$. 
Since we have found one such purification, we have
\begin{equation}
    E_P(A\colon  B)\leq EW(A\colon  B)
\end{equation}
Note that each of the above inequalities is in the $D\to \infty$ limit. 
Combining these two inequalities, we have $E_P(A\colon  B)=EW(A\colon  B)$ up to terms vanishing in the $D\to \infty$ limit. 
It is then also clear that the geometric purification in Refs.~\cite{Takayanagi:2017knl,Nguyen:2017yqw} is the optimal purification to leading order in $D$.

\section{Examples}
\label{sec:eg}

In this section, we provide simple examples of RTNs to demonstrate regions of parameter space where we have proved $E_P(A\colon B)=EW(A\colon B)$. 
While in the continuum limit one generically expects a non-trivial $X$ region as shown in \figref{fig:triway}, for any discrete network we expect a codimension-$0$ region of parameter space where the $X$ region vanishes.

\subsection{1TN}
\label{sub:1TN}
\begin{figure}[ht]
\includegraphics[scale=0.35]{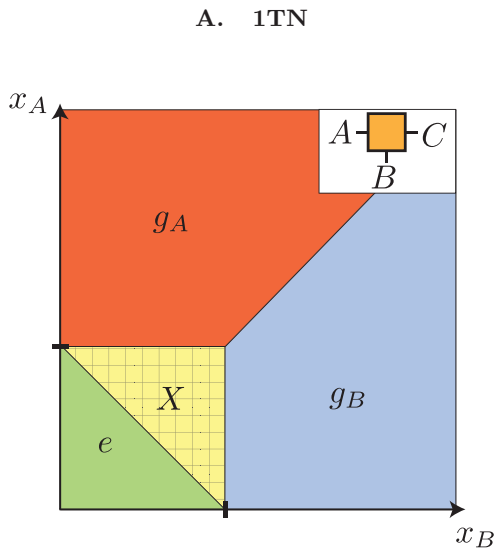}
\caption{The 1TN model (top right), along with its phase diagram labelled by the dominant permutation element in each domain. We have proved $E_P(A\colon B)=EW(A\colon B)$ everywhere except the region marked with squares. \label{fig:1TN}}
\end{figure}

The first example we consider is that of a Haar random tripartite state, represented by a graph with a single vertex and three legs with bond dimensions $d_{A/B/C}$ respectively (see \figref{fig:1TN}).
In this case, the reflected entropy was computed in detail in Ref.~\cite{Akers:2021pvd}. 
We present the phase diagram in \figref{fig:1TN}. 
The phase boundaries at $n=2$ are represented as a function of $x_A=\frac{\log d_A}{\log d_C}$ and $x_B=\frac{\log d_B}{\log d_C}$. 
Apart from the shaded region marking the $X$ domain, we have proved $E_P(A\colon  B)=EW(A\colon  B)$ everywhere else. 
It is also straightforward to read off the optimal purification since we already argued it is given by the geometric purification suggested in Ref.~\cite{Takayanagi:2017knl,Nguyen:2017yqw}.

One may consider a simple deformation of the above model, by changing the maximally entangled legs of the RTN to non-maximally entangled legs. 
Such states have also been useful to model holographic states \cite{Cheng:2022ori}.
In fact, the simplest situation where we add non-maximal entanglement to the $C$ leg results in a state identical to the PSSY model of black hole evaporation \cite{Penington:2019kki}. 
We can thus use the results of Ref.~\cite{Akers:2022max} which computed the reflected entropy in this model. 
The phase diagram turns out to be similar to \figref{fig:1TN} except the shaded region turns out to be larger. 
Thus, non-maximal links do not help in improving the applicability of our result.
We provide some more details on this in Appendix~\ref{app:non_max}.

\subsection{2TN}

\begin{figure}[ht]
\includegraphics[scale=0.35]{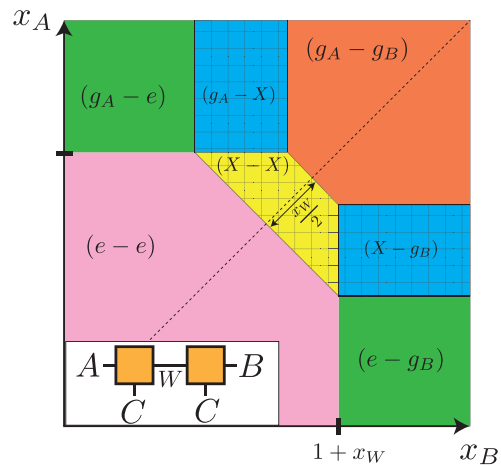}
\caption{The 2TN model (bottom left) and its phase diagram labelled by the dominant permutation element tuple. We have proved $E_P(A\colon B)=EW(A\colon B)$ everywhere except the region marked with squares.}\label{fig:2TN}
\end{figure}

The next simplest network to consider is one where we have two vertices connected by an internal bond labelled $W$ as shown in \figref{fig:2TN}.
For simplicity, the external $C$ bonds are chosen to have identical bond dimension. 

In general, we have the phase diagram shown in \figref{fig:2TN}. 
Again, we see a large codimension-$0$ region of parameter space where our proof applies. 
In fact, motivated by holography, Ref.~\cite{Akers:2022zxr} considered a limit where $x_W = \frac{\log d_W}{\log d_C}\to 0$. 
In this limit, the shaded domains containing the element $X$ vanish at arbitrary $n$. 
Thus, our proof always applies in this limit.

\section{Discussion}
\label{sec:disc}

In this note we have proven $E_P=EW$ for a large class of RTNs. 
Our result relied on the inequality $E_P\geq \frac{1}{2}S_R^{(2)}$ proven as Theorem \ref{thm:lower_bound}. 

Proving the stronger inequality $E_P\geq \frac{1}{2}S_R$ would prove $E_P=EW$ more generally, but this cannot be achieved with our proof technique. 
It would be interesting to check this numerically using the techniques of Ref.~\cite{hauschild2018finding}.

\begin{figure}[t]
\includegraphics[scale=0.35]{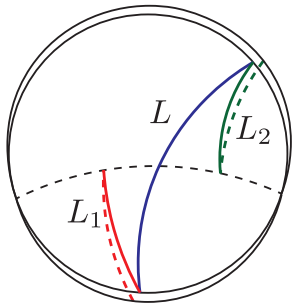}
\caption{A correlation function computed by the geodesic of length $L$ (solid blue) can be compared to the mirror correlation functions analogous to \Eqref{eq:ineq} computed by mirror geodesics (red and green) meeting the RT surface (dashed) orthogonally. Each of the mirror geodesics with length $L_{1,2}$ involves two copies (dashed and solid). It is then clear that $L\geq \frac{L_1+L_2}{2}$.}\label{fig:geom}
\end{figure}

An inequality of the form of \Eqref{eq:ineq} can in fact be proved for heavy local operators in AdS/CFT by using the geodesic approximation and the techniques of computing mirror correlation functions \cite{Faulkner:2018faa} (see \figref{fig:geom}).
In AdS$_3$/CFT$_2$, twist operators are local and can be analytically continued to $n\approx1$. 
Applying the inequality, we would then find $S(AA')\geq \frac{1}{2}S_R(A:B)+\frac{1}{2}S_R(A':B')$ in any geometric purification. 
It would be interesting if this argument can be generalized to non-geometric states, so that we can minimize the LHS and find the strengthened inequality.

\begin{acknowledgments}
PR is supported in part by a grant from the Simons Foundation, and by funds from UCSB. 
CA is supported by the Simons foundation as a member of the It from Qubit collaboration, the NSF grant no. PHY-2011905, and the John Templeton Foundation via the Black Hole Initiative. This material is based upon work supported by the Air Force Office of Scientific Research under award number FA9550-19-1-0360. 
\end{acknowledgments}

\appendix
\section{Non-maximally Entangled RTNs}
\label{app:non_max}

In a standard RTN, the edges are projected onto maximally entangled states. 
These RTN states can be deformed to nearby states by simply changing the entanglement spectrum on the edges. 
One may then ask whether we can prove $E_P=EW$ for a larger class of states by considering such a deformation, and attempting to enlarge the parameter space where the inequality in Theorem~\ref{thm:lower_bound} is saturated.
It turns out the answer is no, and we give an example in this section to highlight the basic issue.

Consider the 1TN model of \secref{sub:1TN} with a non-maximally entangled leg for subregion $C$. 
This state, for a specific choice of spectrum, is identical to that of the PSSY model, an evaporating black hole in JT gravity coupled to end-of-the-world branes with flavour indices entangled with a radiation system \cite{Penington:2019kki}. 
Here, we will not restrict to the PSSY spectrum, and find more generally how this deformation affects the phase diagram of reflected entropy. 

For generality, consider the state $\ket{\rho_{AB}^{m/2}}$, a one parameter generalization of the canonical purification. 
Ref.~\cite{Akers:2022max} computed the entanglement spectrum of $\rho_{AA^*}$ for this state. 
It consists of two features: a single pole of weight $p_d(m)$ and a mound of $\min(d_A^2-1,d_B^2-1)$ eigenvalues with weight $p_c(m)$. 
The weights are given by
\begin{align}\label{eq:pd}
    p_d(m)&= \frac{\tr(\rho_{AB}^{m/2})^2}{d_A d_B \tr(\rho_{AB}^{m})}\\
    p_c(m) &= 1-p_d(m).
\end{align}

Now, we would like to compare the phase diagram of this model with the standard 1TN with maximally entangled legs. 
First note that the transition between $e$ and $X$ in \figref{fig:1TN} is dictated by the location of the entanglement wedge phase transition, which we hold fixed to compare the two models.
Then the remaining question is where the transition from $X$ to $g_A/g_B$ happens.

Consider the region of the phase diagram where $d_A>d_B$.
The transition happens in the connected sector. 
Thus, we have $p_c(m)\approx 1$ and the spectrum of $\rho_{AB}$ is well approximated by the spectrum on the $C$ leg.
Using this, we find that the location of the transition for $S_R^{(2)}$ is given by
\begin{align}
    p_d(m)&=\frac{1}{d_B}.
\end{align}
Using \Eqref{eq:pd}, we then have
\begin{equation}
    (2-m)S_{m/2}-(1-m)S_m=\log d_A,
\end{equation}
where $S_n$ is the $n$th Renyi entropy of the non-maximal spectrum on the $C$ leg.

Then it is clear that at $m=1$, the location of the phase transition is
$x_A = \frac{S_{1/2}}{S_1}\geq 1$. The standard 1TN has a flat spectrum, i.e, $S_n=S_1$ and the transition is at $x_A=1$.
Thus, the shaded region where we cannot prove $E_P=EW$ is larger after deforming the RTN to add non-maximally entangled legs.

As a side note, we would like to mention what happens for $m\geq 2$ where one can use the usual RTN calculation of domain walls with tensions modified by the entanglement spectrum, thus introducing an $m$ dependence \cite{Dong:2021clv,Cheng:2022ori}.
For $m\geq 2$, we have
$x_A = \frac{S_{m/2}}{S_1}-(m-1)\frac{S_{m/2}-S_m}{S_1}\leq 1$ since $S_{m}\leq S_{m/2}\leq S_1$.
Thus, the $X$ region shrinks for $m\geq 2$ after deforming the spectrum on the legs. 
However, as demonstrated above for $m=1$, the naive analytic continuation of the result at $m\geq 2$ fails.

\bibliography{apssamp}

\end{document}